\documentclass[a4paper,UKenglish,cleveref,thm-restate]{lipics-v2021}

\nolinenumbers \graphicspath{{figures/}} 

\usepackage[nobreak,sort]{cite}
\newcommand{\noadjust}{\hspace{1sp}}

\usepackage{amsmath}
\usepackage{stmaryrd}
\usepackage{mathtools}

\usepackage{quiver}

\usetikzlibrary{decorations.markings}
\tikzset{kleisli/.style={ postaction={decorate, decoration={markings, mark= at position 0.5 with {
        \draw circle[radius=1.5pt]; }}
}}}

\usetikzlibrary{calc}
\makeatletter
\newcounter{tikzcdequation}
\renewcommand*\thetikzcdequation{(\arabic{tikzcdequation})}
\newif\iftikzcd@eq@customtag
\tikzcdset{
  number/.code=\pgfqkeys{/tikz/commutative diagrams/number}{
    /tikz/path only,normal tag,normal at,every number,#1,@do},
  number/every number/.code=,
  number/label/.initial=,
  number/tag/.code=\tikzcd@eq@customtagtrue\def\tikzcd@eq@tag{#1},
  number/normal tag/.code=
    \tikzcd@eq@customtagfalse\def\tikzcd@eq@tag{\thetikzcdequation},
  number/normal  at/.code=\def\tikzcd@eq@at{no at},
  number/at/matrix  center/.code=\def\tikzcd@eq@at{matrix center},
  number/at/diagram center/.code=\def\tikzcd@eq@at{diagram center},
  number/at/.code=\pgfkeysifdefined{/tikz/commutative diagrams/number/at/#1/.@cmd}
      {\pgfkeysvalueof{/tikz/commutative diagrams/number/at/#1/.@cmd}\pgfeov}
      {\def\tikzcd@eq@at{at={#1}}},
  number/@at no at/.code=, number/@at matrix  center/.style={
    commutative diagrams/number/@at at=(\tikzcdmatrixname.center)},
  number/@at diagram center/.style={
    commutative diagrams/number/@at at=(current bounding box.center)},
  number/@at at/.style={pos=,to path=\tikztonodes,nodes={at={#1}}},
  number/@do/.style={
    /tikz/every to/.append style/.expanded={
      commutative diagrams/number/@at \tikzcd@eq@at,
      edge node={
        node[commutative diagrams/eq node/.try,
             commutative diagrams/number/@@do\iftikzcd@eq@customtag*={\unexpanded\expandafter{\tikzcd@eq@tag}}\else=\fi
            {\pgfkeysvalueof{/tikz/commutative diagrams/number/label}}]{}}}},
  number/@@do/.style={
    node contents={
      \refstepcounter{tikzcdequation}\thetikzcdequation
      \IfBlankF{#1}{\@ifpackageloaded{cleveref}{\cref@label{#1}}{\ltx@label{#1}}}}},
  number/@@do*/.style n args={2}{
    /utils/exec=\def\thetikzcdequation{(#1)},commutative diagrams/number/@@do={#2}}}
\makeatother
\newcommand{\diagramlabel}[4][]{\ifx\\#1\\\arrow[from=#2, to=#3, number={label=cd:#4}]\else \arrow[from=#2, to=#3, number={label=cd:#4, tag=#1}]\fi }
\crefname{tikzcdequation}{diagram}{diagrams}
\Crefname{tikzcdequation}{Diagram}{Diagrams}
\title{Active Learning of Upward-Closed Sets of Words} 

\author{Quentin Aristote}{Université Paris Cité, CNRS, Inria, IRIF, F-75013,
  Paris,
  France}{quentin.aristote@irif.fr}{https://orcid.org/0009-0001-4061-7553}{} 

\authorrunning{Q. Aristote} 

\Copyright{Quentin Aristote}

\ccsdesc[500]{Theory of computation~Formal languages and automata theory} 

\keywords{active learning, well quasi-orders, Valk-Jantzen lemma, piecewise-testable languages, monoids} 

\category{(Co)algebraic pearl}

\acknowledgements{The author thanks Daniela Petri\c{s}an and anonymous reviewers for their precious comments.} 

\EventEditors{Corina C\^{i}rstea and Alexander Knapp}
\EventNoEds{2}
\EventLongTitle{11th Conference on Algebra and Coalgebra in Computer Science (CALCO 2025)}
\EventShortTitle{CALCO 2025}
\EventAcronym{CALCO}
\EventYear{2025}
\EventDate{June 16--18, 2025}
\EventLocation{University of Strathclyde, UK}
\EventLogo{}
\SeriesVolume{342}
\ArticleNo{16}
 \newcommand{\leftpar}{\mathopen{}\left(}
\newcommand{\rightpar}{\right)\mathclose{}}

\newcommand{\naturals}{\mathbb{N}}

\DeclareMathSymbol{:}{\mathpunct}{operators}{"3A}
\newcommand{\mathand}{\mathrel{\mathrm{and}}}
\newcommand{\suchthat}[2]{\left\{ #1 ~\middle|~ #2 \right\}}

\newcommand{\namedCat}[1]{\mathbf{#1}}

\newcommand{\QO}{\namedCat{QO}}

\newcommand{\down}[1]{\left\downarrow{#1}\right.}
\newcommand{\up}[1]{\left\uparrow{#1}\right.}
\newcommand{\letters}{\Sigma}
\newcommand{\alphabetQO}[1][\letters]{\leftpar{#1},{\preceq}\rightpar}
\newcommand{\words}[1][\letters]{{#1}^*}
\newcommand{\subwordQO}[1][\letters]{\leftpar\words[#1],{\preceq_*}\rightpar}

\newcommand{\atom}{\mathfrak{A}}

\DeclareMathOperator{\Idl}{Idl}

\newcommand{\states}{\mathtt{st}}
\newcommand{\oracle}[1]{\textsc{Oracle}_{#1}}

\newcommand{\abs}[1]{\left| #1 \right|}
 
\begin{document}

\maketitle

\begin{abstract}
  We give a new proof of a result from well quasi-order theory on the
  computability of bases for upward-closed sets of words. This new proof is
  based on Angluin's L* algorithm, that learns an automaton from a minimally
  adequate teacher. This relates in particular two results from the 1980s:
  Angluin's L* algorithm, and a result from Valk and Jantzen on the
  computability of bases for upward-closed sets of tuples of integers.

  Along the way, we describe an algorithm for learning quasi-ordered automata
  from a minimally adequate teacher, and extend a generalization of Valk and
  Jantzen's result, encompassing both words and integers, to finitely generated
  monoids.
\end{abstract}

\section{Introduction}
\label{sec:intro}

\paragraph*{Active learning of automata}
\label{sec:intro:active-learning}

In 1987, Angluin published a seminal work in computational learning theory
\cite{angluinLearningRegularSets1987}. She proved that the minimal complete
deterministic finite-state automaton (DFA) recognizing a regular language $L
\subseteq \words$ can be learned with the help of a \emph{minimally adequate
  teacher (MAT)}, i.e. an oracle able to answer two kinds of queries:
\begin{itemize}
\item \emph{membership queries} -- given a word $w \in \words$, the oracle
  decides whether $w \in L$ --
\item and \emph{equivalence queries} -- given an automaton $A$ which can be
  assumed to be minimal for the language it recognizes, the oracle decides
  whether this language is $L$, and, in the negative case, also outputs a
  counterexample word $w \in \words$ such that $w$ is in $L$ but is not
  recognized by $A$, or conversely.
\end{itemize}
Angluin's algorithm, called $L^*$, computes the minimal DFA
recognizing $L$. Its complexity is polynomial in the size of this minimal
automaton and the maximum length of the counterexamples.

\looseness=-1 The $L^*$ algorithm has proven very versatile over time: it has
been extended to a number of other families of automata, see e.g.
\cite[\textsection 1.2]{vanheerdtCALFCategoricalAutomata2020} for a survey. The
similarities between these extensions and the need to unify them has motivated
several categorical approaches to MAT learning
\cite{colcombetLearningAutomataTransducers2021,vanheerdtCALFCategoricalAutomata2020,urbatAutomataLearningAlgebraic2020,jacobsAutomataLearningCategorical2014,barloccoCoalgebraLearningDuality2019},
some of which have in turn been used to develop new extensions of Angluin's
algorithm
\cite{aristoteActiveLearningDeterministic2024,vanheerdtLearningWeightedAutomata2020}.
In this work we will use yet another extension of the $L^*$ algorithm, for
automata with ordered sets of states, to recover a proof of decidability in well
quasi-order theory.

\paragraph*{The generalized Valk-Jantzen lemma}
\label{sec:intro:generalized-vj-lemma}

Two years before Angluin's 1987 result, Valk and Jantzen proved a seemingly
unrelated theorem. Order $\naturals^k$ component-wise, and, for $z \in
(\naturals_* \cup \{ \infty \})^k$ (where $\naturals_*$ denotes the set of
positive integers), write $I_z = \{(x_1, \ldots, x_n) \in \naturals^k \mid
\forall i, x_i < z_i \}$.

\begin{theorem}[Valk-Jantzen lemma {\cite[Theorem 2.14]{valkResidueVectorSets1985}}]
  \label{thm:vj}
  Fix an upward-closed subset $U \subseteq \naturals^k$. Then, given an oracle
  that decides for any subset $I_z $ with $z \in (\naturals_* \cup \{ \infty
  \})^k$ whether $U$ intersects $I_z$, there is an algorithm that computes a
  minimal \emph{basis} $B$ of $U$ -- a minimal $B \subseteq \naturals^k$ such
  that $U = \suchthat{u \in \naturals^k}{\exists b \in B, u \ge b}$.
\end{theorem}

This result is used, by Valk and Jantzen themselves but also later by other
authors, to show the decidability of a number of problems regarding Petri nets
and similar state machines: see
\cite{goubault-larrecqGeneralizationResultValk2009} for a list of these. In
\textit{ibid.}, Goubault-Larrecq extends \Cref{thm:vj} to a whole grammar of
\emph{well quasi-orders}: this last result is called the \emph{generalized
  Valk-Jantzen lemma}.

Recall that a \emph{quasi-order (QO)}, also called a \emph{preorder}, is a set
equipped with a reflexive and transitive relation. Given a QO $(X,{\le})$, let
$\down{A} = \suchthat{x \in X}{\exists y \in A, x \le y}$ denote the
\emph{downward-closure} of $A \subseteq X$, $\up{A} = \suchthat{x \in X}{\exists
  y \in A, x \ge y}$ its \emph{upward-closure}, and for $x \in X$ let $\down{x}
= \down{\{x\}}$ and $\up{x} = \up{\{x\}}$. An $x \in X$ is \emph{minimal} in
$(X,{\le})$ when it is smaller than every other $y \in \down{x}$. A subset $A$
of a QO $(X,{\le})$ is called \emph{downwards-closed} when $A = \down{A}$, and
upward-closed when $A = \up{A}$.

\looseness=-1
\noindent A QO $(X,{\le})$ is \emph{well (WQO)} when all its upward-closed
subsets $U$ have a finite \emph{basis}, i.e. can be written $U = \up{A}$ for
some finite $A \subseteq X$. Such a basis is said to be \emph{minimal} when it
does not contain any strictly smaller basis of $U$. Finite QOs and well-founded
linear orders such as $\naturals$ are basic examples of WQOs; more involved
examples include finite products of WQOs, equipped with the component-wise
ordering \cite[Lemma 4.7]{goubault-larrecqIdealApproachComputing2020}, or free
monoids thereon [\textit{ibid.}, \textsection 4.4]: if $\alphabetQO$ is a WQO,
then so is the set $\words$ of finite words on $\letters$ equipped with the
\emph{subword ordering}, given inductively by $\varepsilon \preceq_* w$ for all
$w \in \words$, $x \preceq_* y$ when $x \preceq y$ in $\letters$ and $uv
\preceq_* u'v'$ when $u \preceq_* u'$ and $v \preceq_* v'$ -- informally, $w
\preceq_* w'$ if $w'$ can be obtained from $w$ by inserting new letters and
upgrading existing ones. This last example of a WQO is the primary one this work
is concerned with. When $(\letters,=)$ is a finite alphabet the upward-closed
sets of $\subwordQO$ are called \emph{piecewise-testable languages} in formal
language theory: they are the ones specified by unions of regular expressions of
the shape $\words \sigma_1 \words \cdots \words \sigma_n \words$ with $\sigma_i
\in \letters$. It is easy to see that if $\letters$ is still finite but now
equipped with a non-trivial QO structure, the upward-closed sets of
$\subwordQO$ are still regular languages.

The last missing ingredient to understand Goubault-Larrecq's generalization of
Valk and Jantzen's result is the notion of \emph{ideal}: an ideal $I$ is a
non-empty, downwards-closed and \emph{directed} subset, the latter meaning that
for every $x, y \in I$, there is a $z$ greater than both $x$ and $y$ in $I$.
Denote by $\Idl(X,{\le})$ the set of ideals of $(X,{\le})$ ordered by inclusion.
Notice that for every $x \in D$, $\down{x}$ is an ideal, called a
\emph{principal} one.

\begin{theorem}[generalized Valk-Jantzen lemma, {\cite[Theorem
    1]{goubault-larrecqGeneralizationResultValk2009}}]
  \label{thm:generalized-vj}
  Let $(X,{\le})$ be a WQO, and $U$ an upward-closed subset thereof. Given an
  oracle that decides for any $I \in \Idl(X,{\le})$ whether $U$ intersects $I$,
  a minimal finite basis of $U$ is computable\footnotemark[1]{}.
\end{theorem}
\noindent \footnotemark[1]{}Of course this theorem also assumes that $(X,{\le})$
and $\Idl(X,{\le})$ are effective, in a way we do not describe here as all the
instances we consider will indeed be so. The precise requirements are given in
\cite[Definitions 3 and 5]{goubault-larrecqGeneralizationResultValk2009}.

\looseness=-1 For $D = \naturals^k$, \Cref{thm:generalized-vj} is Valk and
Jantzen's original result: the ideals of $\naturals^k$ are exactly the $I_z$'s
for $z$ ranging in $(\naturals_* \cup \{\infty\})^k$ \cite[\textsection 3.1.2
and 4.3]{goubault-larrecqIdealApproachComputing2020}. Consider now the case
$\subwordQO$ where $\alphabetQO$ is a finite quasi-ordered alphabet. Then the
ideals of $\subwordQO$ are the regular languages specified by the following
class of regular expressions, called \emph{$*$-products} \cite[Theorem
4.11]{goubault-larrecqIdealApproachComputing2020}. A $*$-product is any finite
concatenation $P \coloneq \atom_1 \ldots \atom_n$ of \emph{atomic expressions},
where an atomic expression $\atom \coloneq \sigma^? \mid (\Sigma')^*$ is either
some $\sigma^? = \{ \varepsilon \} \cup \suchthat{\sigma' \in \letters}{\sigma'
  \le \sigma}$ for $\sigma \in \letters$, or some $\words[(\letters')] =
\suchthat{\sigma_1 \ldots \sigma_n}{n \in \naturals \mathand{} \forall i,
  \sigma_i \in \letters' }$ for $\letters' \subseteq \letters$ non-empty and
downwards-closed. By abuse of language, we will often speak of a $*$-product to
mean the ideal it specifies. The instance of \Cref{thm:generalized-vj} for
subword orderings now reads:
\begin{restatable}[{\noadjust\cite[Theorem 2]{goubault-larrecqGeneralizationResultValk2009}}]{corollary}{generalizedVJforWords}
  \label{corollary:generalized-vg-words}
  Fix a finite quasi-ordered alphabet $\alphabetQO$ and an upward-closed
  subset $U$ of $\subwordQO$. Given an oracle that decides for any $*$-product
  $P$ whether $U$ intersects $P$, a minimal finite basis of $U$ is computable.
\end{restatable}
\noindent This instance of \Cref{thm:generalized-vj} is especially interesting
in that subword orderings are one of the simplest example of WQOs for which
computing intersections of ideals is not polynomial-time \cite[\textsection
6.3]{goubault-larrecqIdealApproachComputing2020}.

\paragraph*{Structure and contributions}
\label{sec:intro:structure-contributions}

In this work we describe a new algorithm that proves
\Cref{corollary:generalized-vg-words}, using active learning of automata. We
then show that the 1985 result of Valk and Jantzen reduces to this
\Cref{corollary:generalized-vg-words}, and thus to the 1987 result of Angluin.

In \Cref{sec:proof-by-rewriting}, we recall and discuss Goubault-Larrecq's
original proof of \Cref{corollary:generalized-vg-words}. On the way to a new
proof, we describe in \Cref{sec:active-learning-ordered-automata}, as our first
contribution, a new $L^*$-like algorithm
(\Cref{thm:qo-automata-active-learning}) for learning \emph{quasi-ordered
  automata}, a family of automata corresponding to regular languages that are
upward-closed. In \Cref{sec:proof-by-active-learning}, we show as a second
contribution how this algorithm for actively learning quasi-ordered automata can
be used to prove \Cref{corollary:generalized-vg-words}
(\Cref{thm:generalized-vg-words-automaton}). The automata-based approach and the
active learning procedure help alleviate the technical details of
Goubault-Larrecq's original proof. In \Cref{sec:generalized-vj-quotients}, we
finally show how Valk-Jantzen-style learnability of a WQO is stable under
quotients: in particular, because \Cref{corollary:generalized-vg-words} holds
for finitely generated free monoids, it also holds for finitely generated
monoids. As a third contribution, we thus extend the class of WQOs for which
Goubault-Larrecq's general \Cref{thm:generalized-vj} holds, and we retrieve in
particular Valk and Jantzen's original result as $\naturals^k$ is the free
commutative monoid on $k$ letters (\Cref{ex:active-learning-vj}).

\section{A proof by rewriting of regular expressions}
\label{sec:proof-by-rewriting}

In this section we first recall and discuss the proof of the generalized
Valk-Jantzen lemma in the case we focus on in this paper, subword orderings.
This proof is obtained by instantiating \cite[Theorem
1]{goubault-larrecqGeneralizationResultValk2009}'s generic proof in this case.

\generalizedVJforWords*

\begin{proof}
  A minimal finite basis of $U$ can always be derived from a non-minimal one by
  removing redundant elements: if $U = \up{B}$ and $b_1, b_2 \in B$ with $b_1
  \le b_2$, then still $U = \up{(B - \{b_2\})}$. Hence, without loss of
  generality, we now describe a procedure that computes a non-necessarily
  minimal basis.

  Start with $B = \varnothing$. Enumerate all the words in $\words$. For each
  such $w = \sigma_1 \cdots \sigma_n$, check if $w \in U$: because $U$ is
  upward-closed, $w \in U$ if and only if $U$ intersects the $*$-product
  $\down{w} = \sigma_1^? \cdots \sigma_n^?$, and this can be checked by querying
  the oracle. If $w \in U$, add $w$ to $B$, and just move on to the next word
  otherwise.

  By construction, it is always true that $B \subseteq U$. Because $U$ is
  upward-closed, it is hence also always true that $\up{B} \subseteq \up{U} =
  U$. On the other hand, every time a word is added to $B$ and $B$ changes, it
  is possible to check with the help of the oracle whether, conversely, $U
  \subseteq \up{B}$. To check this, notice that it holds if and only if $U \cap
  (\up{B})^c = \varnothing$. But $(\up{B})^c$ is downwards-closed, and it is a
  general fact of WQOs that any downwards-closed set is a finite union of ideals
  \cite[Lemma 2.6 and Proposition
  2.10]{goubault-larrecqIdealApproachComputing2020}, so that $(\up{B})^c = P_1
  \cup \cdots \cup P_n$ for some $n \in \naturals$ and some $*$-products
  $(P_i)_{1 \le i \le n}$. Hence if such $P_i$'s are computable from $B$, one
  can check whether $U \subseteq \up{B}$ by asking the oracle whether $U$
  intersects any of the $P_i$'s.

  Let us now see how to compute these $P_i$'s. Note that the fact that this can
  be done means that $\subwordQO$ has the \emph{effective complement property}
  \cite[Definition 5]{goubault-larrecqGeneralizationResultValk2009}, which is
  one of the effectivity conditions we omitted in the generic
  \Cref{thm:generalized-vj}. We do not give the full construction here as it is
  a bit technical, and only recall the two main steps instead.
  \begin{itemize}
  \item For every $b \in B$, $(\up{b})^c$ can be written as a union of
    $*$-products, each informally obtained by disallowing one letter from $b$ to
    appear, see also \cite[Lemma
    4.14]{goubault-larrecqIdealApproachComputing2020} for more details.
  \item Then $(\up{B})^c = \bigcap_{b \in B} (\up{b})^c$ can be rewritten as the
    union of all intersections of any two $*$-products appearing in the
    decompositions of two different $b, b' \in B$. The intersection of two
    $*$-products can always be rewritten as a union of $*$-products \cite[Lemma
    4.16]{goubault-larrecqIdealApproachComputing2020}, and hence $(\up{B})^c$
    can be rewritten into a union of $*$-products.
  \end{itemize}
  All in all, we always have that $\up{B} \subseteq U$, and every time $B$
  changes we can check whether $U \subseteq \up{B}$ and thus whether $U =
  \up{B}$. This must be the case after having considered a finite number of
  words: $U$ has a finite basis $B'$ as $\subwordQO$ is a WQO, and it must be
  that $B' \subseteq B$ once all of $B'$'s words have been considered.
\end{proof}

The main appeal of the above proof is its genericity and modularity: it works in
fact for any WQO that can be shown to be sufficiently effective, in particular
having the aforementionned effective complement property -- an effective way to transform
the complement of an upward-closed set into a union of ideals.

What hinders the algorithm from being polynomial-time is the rewriting of
regular expressions into unions of $*$-products (which we did not describe in
detail), and the brute-force enumeration of words. While we do not claim to give
another algorithm with better asymptotic complexity, this brute-force
enumeration seems particularly inefficient and it is natural to try to improve
it, with heuristics for instance.

Another concern that could be raised is the use of bases to represent languages.
These are usually represented by automata instead, and this can be exponentially
more compact: consider for instance $U$ to be the set of all words of length at
least $n$, whose minimal basis contains $\abs{\letters}^n$ words but for which
an $(n+1)$-state automaton is given in \Cref{fig:qo-automata:compact}. Note that
this is not always the case: the upward-closed language $\up{\{ \sigma_1,
  \sigma_2^2, \ldots, \sigma_n^n \}}$ has a much more compact representation in
its minimal basis than in its minimal automaton. For instance in the case $n =
3$, the minimal automaton for that language is depicted in
\Cref{fig:qo-automata:example}. In general, converting between bases and
automata may thus be costly, and it is all the more interesting to also have an
algorithm for \Cref{corollary:generalized-vg-words} that works with automata
instead of bases.

\begin{figure}[h]
  \caption{Examples of quasi-ordered automata (\Cref{def:qo-automaton}). All the
    automata we consider are assumed to be complete: if a transition $q
    \xrightarrow{\sigma} \cdot$ is not depicted, then it is assumed to be
    looping, i.e. $\delta_\sigma(q) = q$. Initial states are depicted with an
    incoming arrow, accepting states with a double circle. The quasi-order on
    states is not explicitly depicted, but it always goes from left to right.}
  \label{fig:qo-automata}

  \begin{subfigure}{\textwidth}
    \caption{An automaton recognizing the language
      $\up{\leftpar\letters^n\rightpar}$.}
    \label{fig:qo-automata:compact}
    \includegraphics[scale=1]{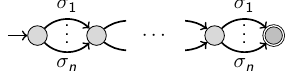}
  \end{subfigure}

  \bigskip

  \begin{subfigure}{\textwidth}
    \caption{An automaton recognizing the language $\up{\{a,bb,ccc\}}$.}
    \label{fig:qo-automata:example}
    \vspace{-15pt}
    \includegraphics[scale=1]{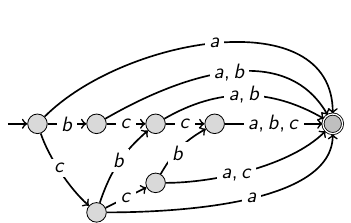}
  \end{subfigure}
\end{figure}

In the rest of this work we address these two concerns by shedding the light of
active learning of automata on the above proof of
\Cref{corollary:generalized-vg-words}. It is quite reasonable to believe that
MAT-based learning should be related to this result. Indeed, looking closely at
the proof, we see that the oracle is really used in two distinct ways: for
checking whether a word belongs to $U$ -- which is reminescent of the membership
queries to an MAT -- and for checking whether a set $B$ is a basis of $U$ --
which is reminescent of the equivalence queries to an MAT. While Angluin's $L^*$
algorithm is only concerned with DFAs and does not involve any QOs, the recent
advances providing categorical frameworks for active learning of automata are a
strong hint that this kind of learning is versatile enough to also be done in a
quasi-ordered setting.

\newpage
\section{Active learning of quasi-ordered automata}
\label{sec:active-learning-ordered-automata}

We now describe how active learning of automata can indeed be done in a
quasi-ordered setting. Because we are only trying to learn those languages that
are upward-closed, DFAs are too general for our purpose. We will focus on a
more restricted family of automata:

\begin{definition}[quasi-ordered automaton]
  \label{def:qo-automaton} Let $\alphabetQO$ be a finite quasi-ordered alphabet.
  A \emph{quasi-ordered automaton} on $\alphabetQO$ is a tuple
  $((Q,{\sqsubseteq}),q_0, \delta, F)$ where:
  \begin{itemize}
  \item $(Q,{\sqsubseteq})$ is a finite QO of \emph{states};
  \item $q_0 \in Q$ is the \emph{initial state};
  \item $\delta : \alphabetQO \times (Q,{\sqsubseteq}) \to
    (Q,{\sqsubseteq})$ is a monotone \emph{transition function} such that $q
    \sqsubseteq \delta(\sigma,q)$ for every $(\sigma,q) \in \letters \times Q$;
  \item $F \subseteq Q$ is an upward-closed subset of \emph{accepting states}.
  \end{itemize}
  The \emph{language recognized by} a quasi-ordered automaton
  $((Q,{\sqsubseteq}),q_0,\delta,F)$ is the language recognized by its
  \emph{underlying DFA} $(Q,q_0,\delta,F)$.
\end{definition}

\noindent Note that the quasi-order on $\alphabetQO \times (Q,{\sqsubseteq})$ is
given by $(\sigma, q) \le (\sigma', q')$ if and only if $\sigma \preceq \sigma'$
and $q \sqsubseteq q'$, so that $\delta$ being monotone means: if $\sigma
\preceq \sigma'$ and $q \sqsubseteq q'$, then $\delta(\sigma,q) \sqsubseteq
\delta(\sigma',q')$. $\delta$ extends to a monotone function $\delta^*:
\subwordQO \times (Q,{\sqsubseteq}) \to (Q,{\sqsubseteq})$, and we also write
$\delta^*_w: (Q,{\sqsubseteq}) \to (Q,{\sqsubseteq})$ for the monotone function
$q \mapsto \delta(w,q)$.

\begin{remark}
  The category-enclined reader may notice that quasi-ordered automata are
  equivalently certain $\QO$-enriched functors, where $\QO$ is the the category
  of quasi-orders and monotone functions between them. Let indeed
  $\namedCat{I}_{\subwordQO}$ be the $\QO$-enriched category freely generated by
  the $\QO$-enriched quiver below (the QO of edges $\states \to \states$ is
  $\alphabetQO$), with the additional requirement that $\varepsilon: \states \to
  \states$ be smaller than every $w: \states \to \states$
\[\begin{tikzcd}
      {\mathtt{in}} & {\states} & {\mathtt{out}}
      \arrow["\triangleright", from=1-1, to=1-2]
      \arrow["{\sigma \in \alphabetQO}"', from=1-2, to=1-2, loop, in=125, out=55, distance=10mm]
      \arrow["\triangleleft", from=1-2, to=1-3]
    \end{tikzcd}\] Then quasi-ordered automata are exactly enriched functors
  $\namedCat{I}_{\subwordQO} \to \QO$ that send $\mathtt{in}$ on the singleton
  QO and $\mathtt{out}$ on $2 = \{ \bot < \top \}$. The active learning
  algorithm of \Cref{thm:qo-automata-active-learning} below could thus be
  obtained by extending the functorial framework of
  \cite{colcombetAutomataMinimizationFunctorial2020,colcombetLearningAutomataTransducers2021}
  to enriched functors, but for the sake of conciseness we do not delve in this
  as we would only end up using the $\QO$-enriched case here.

  It seems the full expressiveness of enriched functors is needed to model
  quasi-ordered automata. Classical functors $\namedCat{I}_{\subwordQO} \to \QO$
  indeed only encompass that each $\delta_\sigma: (Q,{\sqsubseteq}) \to
  (Q,{\sqsubseteq})$ be monotone, but not that $\delta: \alphabetQO \times
  (Q,{\sqsubseteq}) \to (Q,{\sqsubseteq})$ be monotone. The latter condition can
  be encompassed by considering instead pointed $\QO$-coalgebras for the functor
  $2 \times \QO(\alphabetQO,-)$, but these in turn need not satisfy that $q
  \sqsubseteq \delta(\sigma, q)$ for every $(\sigma,q) \in \Sigma \times Q$,
  which is crucial in proving that $\delta^*: \subwordQO \times
  (Q,{\sqsubseteq}) \to (Q,{\sqsubseteq})$ is monotone, itself key in proving
  that quasi-ordered automata recognize upward-closed languages
  (\Cref{lemma:qo-automata-recognize-upward-closed-languages} below).
\end{remark}

Quasi-ordered automata are indeed a sensible family of automata to consider when
working with upward-closed languages:

\begin{lemma}
  \label{lemma:qo-automata-recognize-upward-closed-languages} The language
  recognized by a quasi-ordered automaton on $\alphabetQO$
  is upward-closed in $\subwordQO$.
\end{lemma}

\begin{proof}
  Let $L$ be a language recognized by a quasi-ordered automaton
  $((Q,{\sqsubseteq}),q_0, \delta, F)$ on $\alphabetQO$, and consider some $w
  \preceq_* w' \in \words$. Because $\delta^* : \subwordQO \times
  (Q,{\sqsubseteq}) \to (Q,{\sqsubseteq})$ is monotone, $\delta^*(w,q_0)
  \sqsubseteq \delta^*(w',q_0)$. In particular if $\delta^*(w,q_0) \in F$ and $w
  \in L$, then because $F$ is upward-closed $\delta^*(w',q_0) \in F$ and $w'
  \in L$ as well: $L$ is upward-closed.
\end{proof}

\begin{proposition}
  \label{prop:minimal-DFA-quasi-ordered-iff-language-upward-closed}
  Let $\alphabetQO$ be a finite quasi-ordered alphabet. Given a reachable DFA
  $A$ on $\Sigma$, it can be decided whether $A$ recognizes a language that is
  upward-closed in $\subwordQO$. If this is the case, $A$ is in fact a
  quasi-ordered automaton and the corresponding QO is computable; otherwise two
  words $w \preceq_* w'$ in $\subwordQO$ such that $A$ accepts $w$ but not $w'$
  are computable. All of this can moreover be done in polynomial time in the
  size of $A$.
\end{proposition}

\begin{proof}
  Let $A = (Q,q_0,\delta,F)$. Define an order on $Q$ as follows:
  \[ q \sqsubseteq q' \iff \forall w \in \words, \delta^*_w(q) \in F \Rightarrow
    \delta^*_w(q') \in F \] Deciding whether $q \sqsubseteq q'$ amounts to
  checking whether the language $L_q$ recognized when making $q$ initial is
  included in the language $L_{q'}$ recognized when making $q'$ initial. It is
  standard that this can be done in polynomial time, see e.g. \cite[Theorem
  2.4]{lodingAutomataFiniteTrees2021}: compute an automaton recognizing $L_q
  \cap L_{q'}^c$, and decide whether it recognizes an empty language. If this is
  not the case, an additional $w \in L_q \cap L_{q'}^c$, i.e. such that
  $\delta^*_w(q) \in F$ but $\delta^*_w(q') \notin F$, is computable in
  polynomial time.

  \looseness=-1
  If this order makes $A$ into a quasi-ordered automaton, then by
  \Cref{lemma:qo-automata-recognize-upward-closed-languages} $A$ recognizes an
  upward-closed language. Otherwise, we show next that the language recognized
  by $A$ cannot be upward-closed, and construct two words $w \preceq_* w'$ such
  that $A$ accepts $w$ but not $w'$.

  Here we use that $A$ is reachable so that for every $q \in Q$, there is a $u
  \in \words$ such that $q = \delta^*_u(q_0)$. If ${\sqsubseteq}$ fails to
  make into $A$ into a quasi-ordered automaton, then either
  \begin{itemize}
  \item $F$ is not upward-closed: this is in fact not possible, as $q \in F$
    means $\delta_\varepsilon(q) \in F$ and thus, if $q \sqsubseteq q'$, $q' =
    \delta_\varepsilon(q') \in F$ as well;
  \item there are some $q \sqsubseteq q' \in Q$, $\sigma \in \letters$ such that
    $\delta(\sigma,q) \not\sqsubseteq \delta(\sigma,q')$: this is not possible
    neither, since $q \sqsubseteq q'$ implies that for every $v \in \words$,
    $\delta^*_v(\delta_\sigma(q)) \in F \Rightarrow
    \delta^*_v(\delta_\sigma(q')) \in F$, and thus that $\delta_\sigma(q)
    \sqsubseteq \delta_\sigma(q')$;
  \item there are some $q \in Q$, $\sigma \preceq \sigma' \in \letters$ such
    that $\delta(\sigma, q) \not\sqsubseteq \delta(\sigma',q)$: taking $u,v \in
    \words$ respectively such that $\delta^*_u(q_0) = q$ and $\delta^*_v
    (\delta_\sigma(q)) \in F$ but $\delta^*_v(\delta_\sigma(q)) \notin F$, then
    $A$ accepts $u \sigma v$ but not $u \sigma' v$, yet $u \sigma v \preceq_* u
    \sigma' v$;
  \item there are some $q \in Q$, $\sigma \in \letters$ such that $q
    \not\sqsubseteq \delta_\sigma(q)$: taking $u,v \in \words$ respectively such
    that $\delta^*_u(q_0) = q$ and $\delta^*_v(q) \in F$ but
    $\delta^*_v(\delta_\sigma(q)) \notin F$, then $A$ accepts $uv$ but not
    $u \sigma v$, yet $uv \preceq_* u \sigma v$. \qedhere
  \end{itemize}
\end{proof}

\looseness=-1 Given an upward-closed language $U$, we call \emph{minimal
  quasi-ordered automaton} the minimal DFA recognizing $U$, equipped with the
quasi-order described in the proof of
\Cref{prop:minimal-DFA-quasi-ordered-iff-language-upward-closed}. The minimal
quasi-ordered automaton is state-minimal among all quasi-ordered automata
recognizing $U$, and it has a very specific structure, simple examples of which
are given by the automata in \Cref{fig:qo-automata}:

\begin{lemma}[structure of the minimal quasi-ordered automata]
  \label{lemma:minimal-qo-automata-structure}
  \looseness=-1 Let $A = ((Q,{\sqsubseteq}),q_0, \delta, F)$ be the minimal
  quasi-ordered automaton for an upward-closed language $U$. Then
  \begin{enumerate}
  \item $(Q,{\sqsubseteq})$ is not just a quasi-order, but an order, so that,
    except for transitions $q \xrightarrow{\sigma} q$ that loop on a state, the
    underlying graph of the automaton is acyclic;
  \item the initial state is the minimum of $(Q,{\sqsubseteq})$;
  \item unless $U$ is empty, an accepting state is reachable from any state;
  \item there is at most one accepting state, and it is the maximum of
    $(Q,{\sqsubseteq})$.
  \end{enumerate}
\end{lemma}

\begin{proof} We implicitly use the Myhill-Nerode theorem, and the quasi-order
  constructed in the proof of
  \Cref{prop:minimal-DFA-quasi-ordered-iff-language-upward-closed}.
  \begin{enumerate}
  \item If two states $q$ and $q'$ are such that $q \sqsubseteq q'$ and $q'
    \sqsubseteq q$, then for every $w \in \words$, $\delta^*_w(q) \in F \iff
    \delta^*_w(q) \in F$. In a minimal DFA, this can never be true of two distinct
    states (otherwise the two states could be merged).
  \item For every $q = \delta^*_u(q_0) \in Q$ and $v \in \words$ such that
    $\delta^*_v(q_0) \in F$, $v \in U$. $U$ is upward-closed, hence $uv \succeq
    v$ is in $U$ as well and $\delta^*_v(q) = \delta^*_{uv}(q_0) \in F$: $q_0
    \sqsubseteq q$.
  \item Pick some $q \in Q$ and $v \in U$. The argument for item 2 above
    applies.
  \item Because the language recognized by the automaton is upward-closed, any
    accepting state $q$ is such that $\delta^*_w(q) \in F$ for every
    $w \in \words$. Hence the uniqueness of such an accepting state, and its
    maximality. \qedhere
  \end{enumerate}
\end{proof}

We now consider the problem of learning such a minimal quasi-ordered automaton
from an MAT. For a fixed upward-closed subset $U$ of $\subwordQO$, suppose given
an oracle able to answer two kinds of queries:
\begin{itemize}
\item \emph{membership queries} -- given a word $w \in \words$, the oracle
  decides whether $w \in U$ --
\item and \emph{equivalence queries} -- given a quasi-ordered automaton $A$ that
  is minimal for the upward-closed language it recognizes, the oracle decides
  whether this language is $U$, and, in the negative case, also outputs a
  counterexample word $w \in \words$ such that $w$ is in $U$ but is not
  recognized by $A$, or conversely.
\end{itemize}
Can $U$ be learned by querying this oracle a finite number of times? We answer
this in the positive:

\begin{theorem}[active learning of quasi-ordered automata]
  \label{thm:qo-automata-active-learning}
  Suppose given an MAT as above. There is an algorithm that computes the
  minimal quasi-ordered automaton recognizing $U$, and its complexity is
  polynomial in the size of this minimal automaton and the maximum length of a
  counterexample outputted by the oracle.
\end{theorem}

\begin{proof}
  We prove this by giving a polynomial-time reduction to the active learning
  problem for DFAs, where $U$ is only seen as a regular language, the complexity
  result being that of the classical active learning algorithm of Angluin
  \cite{angluinLearningRegularSets1987}. We therefore explain how to implement
  an oracle in the classical setting, call it $\oracle{U}$, from the oracle in
  the quasi-ordered setting, call it $\oracle{(U,{\preceq_*})}$: one can then
  learn the minimal DFA recognizing the regular language $U$ in polynomial-time
  using the classical active learning algorithm, and then equip this minimal DFA
  with a quasi-order in polynomial-time as explained in
  \Cref{prop:minimal-DFA-quasi-ordered-iff-language-upward-closed}.

  First, any membership query to $\oracle{U}$ can simply be forwarded to
  $\oracle{(U,{\preceq_*})}$: when it comes to membership queries, the two
  oracles have the same computing power. Consider now a DFA $A$ minimal for the
  language it recognizes, given to $\oracle{U}$ for an equivalence query. By
  \Cref{prop:minimal-DFA-quasi-ordered-iff-language-upward-closed}, one can then
  check in polynomial-time whether $A$ is in fact a quasi-ordered automaton:
  \begin{itemize}
  \item if it is the case we also get the QO structure on $A$, and can then
    forward $A$ to $\oracle{(U,{\preceq_*})}$ for an equivalence query;
  \item otherwise, we get some $w \preceq_* w'$ such that $A$ recognizes $w$ but
    not $w'$, and one of these two words must then be a counterexample to $A$
    recognizing $U$. \qedhere
  \end{itemize}
\end{proof}

\section{A proof by active learning of automata}
\label{sec:proof-by-active-learning}

We now give a new proof of \Cref{corollary:generalized-vg-words}, by reducing it
to the problem of actively learning quasi-ordered automata.

\begin{theorem}
  \label{thm:generalized-vg-words-automaton}
  Fix a finite quasi-ordered alphabet $\alphabetQO$ and an upward-closed subset
  $U$ of $\subwordQO$. Given an oracle that decides for any $*$-product $P$
  whether $U$ intersects $P$, a minimal quasi-ordered automaton recognizing $U$
  is computable.
\end{theorem}

\begin{proof}
  We reduce this problem to that of actively learning quasi-ordered automata,
  for which we gave a polynomial time algorithm in
  \Cref{thm:qo-automata-active-learning}. In other words, we describe how to
  implement an MAT in the quasi-ordered setting by checking whether $U$
  intersects some well-chosen $*$-products.

  Answering a membership query for a $w = \sigma_1 \cdots \sigma_n \in \words$
  can be done by checking whether $U$ intersects $\sigma_1^? \cdots \sigma_n^?$:
  because $U$ is upward-closed, this is the case if and only if $w \in U$.
  Consider now some minimal quasi-ordered automaton $A = ((Q,{\sqsubseteq}),q_0,
  \delta, F)$ submitted to the MAT for an equivalence query.
  We decide whether $A$ recognizes $U$ in two steps, implicitly using the
  properties of $A$ described in \Cref{lemma:minimal-qo-automata-structure} and
  assuming $A$ does not recognize the empty language.
  \begin{itemize}
  \item First, we check whether $U$ contains the language recognized by $A$. For
    this, enumerate all the paths $q_0 \xrightarrow{\sigma_1} \cdots
    \xrightarrow{\sigma_n} q_n$ in $A$ that go from the initial state $q_0$ to
    the accepting state $q_n \in F$, that are \emph{strictly increasing} -- $q_0
    \sqsubset q_1 \sqsubset \cdots \sqsubset q_{n-1} \sqsubset q_n$ -- and that
    are \emph{minimal} -- for every $q_i \xrightarrow{\sigma_i} q_{i+1}$,
    $\sigma_i$ is minimal among the letters $\sigma$ such that
    $\delta_\sigma(q_i) = q_{i+1}$ (see
    \Cref{fig:strictly-increasing-paths:to-accepting} for an example within the
    automaton of \Cref{fig:qo-automata:example}). Check whether $\sigma_1 \cdots
    \sigma_n \in U$ with a membership query. If this is not the case, $\sigma_1
    \cdots \sigma_n$ is a counterexample. Otherwise, if no path yields such a
    counterexample, then all the words accepted by $A$ are in $U$ as well: if
    $A$ accepts $w$, then, for one of the paths $q_0 \xrightarrow{\sigma_1}
    \cdots \xrightarrow{\sigma_n} q_n$ enumerated above, $w \succeq_* \sigma_1
    \cdots \sigma_n$ and $w \in U$ as $U$ is upward-closed.
  \item Second, we check whether $U$ is contained in the language recognized by
    $A$, i.e. whether every word not accepted by $A$ is not in $U$ neither. For
    this, enumerate all the strictly increasing paths $q_0
    \xrightarrow{\sigma_1} \cdots \xrightarrow{\sigma_n} q_n$ in $A$ that go
    from the initial state $q_0$ to a non-accepting state $q_n \notin F$ one
    transition away from the accepting state (see
    \Cref{fig:strictly-increasing-paths:to-non-accepting} for an example within
    the automaton of \Cref{fig:qo-automata:example}). Let $\letters_i \subseteq
    \letters$ be the subset of these letters $\sigma$ for which
    $\delta_\sigma(q_i) = q_i$. Then any word $w$ that is not accepted by $A$
    must be in one of the $*$-products $\letters_0^* \sigma_1^? \letters_1^*
    \cdots \letters_{n-1}^* \sigma_n^? \letters_n^*$ corresponding to one of
    these paths: the run for $w$ in $A$ must stop in a non-accepting state, and,
    since $ww'$ is accepted by $A$ for any $w'$ accepted by $A$, the run for $w$
    can be completed into a run that ends in the accepting state. It is thus
    enough to check whether $U$ intersects any of these $*$-products. If this is
    the case for one of these $*$-products $P$, then enumerating the words in
    $P$ will end up producing a counterexample $w \in \Sigma^*$ such that $w \in
    U$ but $w$ is not accepted by $A$. \qedhere
  \end{itemize}
\end{proof}

\begin{figure}[h]
  \centering
  \caption[Strictly increasing paths.]{Strictly increasing paths in the
    automaton of \Cref{fig:qo-automata:example}.}
  \label{fig:strictly-increasing-paths}
  \begin{subfigure}{.4\textwidth}
    \caption{Two strictly increasing paths to the accepting state. They are
      minimal as soon as $a \preceq b \Rightarrow b \preceq a$ in $\alphabetQO$.}
    \label{fig:strictly-increasing-paths:to-accepting}
    \vspace{-15pt}
    \includegraphics{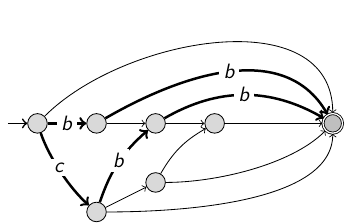}
  \end{subfigure}
  \qquad
  \begin{subfigure}{.4\textwidth}
    \caption{Two strictly increasing paths to non-accepting states one
      transition away from the accepting state.}
    \label{fig:strictly-increasing-paths:to-non-accepting}
    \vspace{-15pt}
    \includegraphics{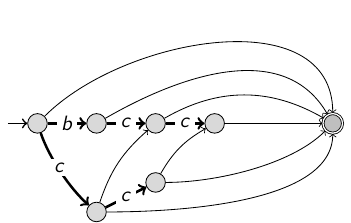}
  \end{subfigure}
\end{figure}

 Note that a finite basis of $U$ can then be obtained from its minimal
 quasi-ordered automaton by enumerating words labelling minimal strictly
 increasing paths. 

 This new proof partially addresses the two concerns highlighted in
 \Cref{sec:proof-by-rewriting}: upward-closed languages are now represented
 using automata, and not finite bases, and the MAT learning procedure (that we
 leave as a black box) takes care of efficiently enumerating words and
 constructing guesses to submit for an equivalence query. The reduction we
 provide is where the non-polynomial-time part happens: what now makes the
 algorithm exponential-time is the enumeration of paths to the accepting state
 in a minimal quasi-ordered automaton, and the brute-force enumeration of words
 in a $*$-product to produce a counterexample.

\section{The Valk-Jantzen lemma for quotients of WQOs}
\label{sec:generalized-vj-quotients}

We now move on from quasi-orders on the finitely generated free monoid $\words$
of words, to quasi-orders on finitely generated monoids. We thus consider a
monoid $M$ equipped with a quasi-order ${\preceq_M}$, and generated by a finite
quasi-ordered alphabet $\alphabetQO$: we assume there is a surjective
homomorphism of monoids $e_M: \words \to M$ that is also a monotone function
$\subwordQO \to (M,{\preceq_M})$. The surjectivity and monotonicity of $e_M$
automatically makes $(M,{\preceq_M})$ into a WQO, and $e_M$ can be used to
describe the ideals of $(M,{\preceq_M})$:

\begin{lemma}[ideals of a quasi-ordered monoid] \label{lemma:ideals-monoid} The
  ideals of a quasi-ordered monoid $(M,{\preceq_M})$ generated by $\alphabetQO$
  are exactly the downward-closures of images of languages specified by
  $*$-products:
  \[ \Idl(M,{\preceq_M}) = \suchthat{\down{e_M[I]}}{I \in \Idl\subwordQO}\]
\end{lemma}

\begin{proof}
  This result is an instance of \cite[Proposition
  5.1]{goubault-larrecqIdealApproachComputing2020}, but the phrasing differs
  slightly so we repeat the proof here for completeness.

  On the one hand, if $I$ is an ideal of $\subwordQO$ (a language described by a
  $*$-product) then $\down{e_M[I]}$ is an ideal of $(M,{\preceq_M})$. Indeed,
  $e_M[I]$ is non-empty and directed because $I$ is: for every $x, y \in
  e_M[I]$, write $x = e_M(x')$ and $y = e_M(y')$ with $x',y' \in I$; then there
  is a $z' \in I$ such that $z' \succeq_* x'$ and $z' \succeq_* y'$, so that
  $e_M(z') \in e_M[I]$, $e_M(z') \succeq_M e_M(x')$ and $e_M(z') \succeq_M
  e_M(y')$. $\down{e_M[I]}$ is thus also non-empty and directed, and it is of
  course downwards-closed: it is an ideal.

  On the other hand, we show that every ideal of $(M,{\preceq_M})$ is of the
  shape $\down{e_M[I]}$ for some $I \in \Idl\subwordQO$. To see this, we use the
  following characterization of ideals in a WQO: a downwards-closed subset $D$
  can always be written as a finite union $D = \bigcup_{i = 1}^n I_i$ of ideals
  \cite[Lemma 2.6(2)]{goubault-larrecqIdealApproachComputing2020}, and $D$ is an
  ideal if and only if it is join-irreducible, i.e. for every finite family of
  downwards-closed subsets $D_1, \ldots, D_n$, if $D = \bigcup_{i = 1}^n D_i$
  then $D = D_i$ for some $1 \le i \le n$ \cite[Lemma
  2.5]{goubault-larrecqIdealApproachComputing2020}. Consider thus a
  downwards-closed subset $D$ of $(M,{\preceq_M})$. Then $e_M^{-1}[D]$ is
  downwards-closed in $\subwordQO$, hence $e_M^{-1}[D] = \bigcup_{i = 1}^n I_i$ for
  some $I_1, \ldots, I_n \in \Idl\subwordQO$. But then
  \[ D = \down{e_M\mathopen{}\left[e_M^{-1}[D]\right]\mathclose{}} =
    \down{e_M\mathopen{}\left[ \bigcup_{i = 1}^n I_i \right]\mathclose{}} =
    \down{\bigcup_{i = 1}^n e_M[I_i]} = \bigcup_{i = 1}^n \down{e_M[I_i]} \] hence
  when $D$ is an ideal it must be equal to $\down{e_M[I_i]}$ for some $1 \le i \le
  n$.
\end{proof}

If $P$ is a $*$-product that specifies an ideal $I$ of $\subwordQO$, write by
abuse of notation $e_M[P]$ for the downard-closed set $e_M[I]$ of
$(M,{\preceq_M})$. Finitely generated quasi-ordered monoids also enjoy a
Valk-Jantzen-like result:

\begin{theorem} \label{thm:generalized-vj-monoids} Fix a finite quasi-ordered
  alphabet $\alphabetQO$ and an upward-closed subset $U$ of a quasi-ordered
  monoid $(M,{\preceq_M})$ generated by $\alphabetQO$. Given an oracle that
  decides for any $*$-product $P$ whether $U$ intersects $\down{e_M[P]}$, a
  finite basis of $U$ is computable as soon as $e_M$ is.
\end{theorem}

\begin{proof}
  We reduce this problem to that of computing a finite basis of the
  upward-closed subset $e_M^{-1}[U]$ of $\subwordQO$. For this it is enough, by
  \Cref{corollary:generalized-vg-words}, to decide, given a $*$-product $P$,
  whether $e_M^{-1}[U]$ intersects $P$. But $e_M^{-1}[U]$ intersects $P$ if and
  only if $U$ intersects $\down{e_M[P]}$: if $x \in e_M^{-1}[U] \cap P$, then
  $e_M(x) \in U \cap \down{e_M[P]}$; conversely, if $x \in U \cap \down{e_M[P]}$,
  then there is a $y \in P$ such that $x \le e_M(y)$, and of course $e_M(y) \in
  U \cap \down{e_M[P]}$ as $U$ is upward-closed, so that $y \in e_M^{-1}[U] \cap
  P$.

  Once a minimal finite basis $B$ of $e_M^{-1}[U]$ is computed, $e_M[B]$ is a
  finite basis of $U$. Note that $e_M[B]$ need not be a minimal basis, and need
  not be effectively reducable to a minimal one, as that would require the
  quasi-order ${\preceq_M}$ on $M$ to be decidable, an assumption we have not
  made.
\end{proof}

\Cref{thm:generalized-vj-monoids} is not just an instance of Goubault-Larrecq's
\cite[Theorem 1]{goubault-larrecqGeneralizationResultValk2009}. Indeed, the
latter result requires an effectivity assumption (namely, the effective
complement property we mentioned in the proof in
\Cref{sec:proof-by-rewriting}) which need not be true of any quotient of the
free monoid, as discussed in \cite[Theorem
5.2]{goubault-larrecqIdealApproachComputing2020}.

\looseness=-1 We restricted ourselves to monoids, but the careful reader may
have noticed that nowhere in the proofs of
\Cref{lemma:ideals-monoid,thm:generalized-vj-monoids} we actually use the fact
that $\subwordQO$ is a subword ordering or that $M$ is a monoid. We could in
fact have stated \Cref{thm:generalized-vj-monoids} more generally, and it would
then informally read as follows. \textit{Fix a WQO $(X,{\le})$ for which
  \Cref{thm:generalized-vj} holds, a monotone surjection $e: (X,{\le}) \to
  (Y,{\le'})$, and an upward-closed subset $U$ of $(Y,{\le'})$. Then
  $(Y,{\le'})$ is a WQO and, given an oracle that decides for any $I \in
  \Idl(X,{\le})$ whether $U$ intersects $\down{e[I]}$, a finite basis of $U$ is
  computable}. Again, we do not make this any more formal to avoid delving in
the topic of effective representations of WQOs and their ideals.

\begin{example} \label{ex:active-learning-vj}
  Let $(\letters,{=})$ be an alphabet with $k$ letters. Then $\naturals^k$, also
  known as the free commutative monoid with $k$ generators, is a quotient of
  $\words$: the surjection $\abs{-}: \words \to \naturals^k$, written $\abs{w} =
  (\abs{w}_1,\ldots,\abs{w}_k)$ counts the number of appearances of each letter
  in a word. This surjection is monotone from the subword ordering on $\words$
  to the component-wise ordering on $\naturals^k$. Given a $*$-product
  $\letters'_0 \sigma_1^? \letters'_1 \cdots \letters'_{n-1} \sigma_n^?
  \letters'_n$, the downward-closure of its image by $\abs{-}$ is the subset of
  elements $(x_1, \ldots, x_k) \in \naturals^k$ such that, for every $1 \le i
  \le k$,
  \begin{itemize}
  \item $x_i \in \naturals$ if the $i$-th letter appears in $\bigcup_{j = 0}^n
    \letters'_j$;
  \item $x_i \le \abs{\sigma_1 \cdots \sigma_n}_i$ otherwise.
  \end{itemize}
  Hence the corresponding instance of \Cref{thm:generalized-vj-monoids} is Valk
  and Jantzen's original 1985 result, except we have now proved it by reducing
  it to Angluin's 1987 active learning algorithm.
\end{example}

Representing elements of $\naturals^k$ as $k$-letter words is of course not very
efficient, and what is interesting in \Cref{ex:active-learning-vj} is how it
reduces the result of Valk and Jantzen to that of Angluin. But for other non-trivial
monoids, using word representations may still be a sensible choice: a trace
monoid where only a few pairs of letters are allowed to commute, for instance,
is very close to a monoid of words.

\section{Conclusion and further work}
\label{sec:conclusion}

In this work we have provided a new proof, relying on active learning of
automata, of the generalized Valk-Jantzen lemma for subword orderings. This
reduces in particular the original 1985 result by Valk and Jantzen for
$\naturals^k$ to Angluin's 1987 $L^*$ algorithm. Along the way, we have
described a new MAT-based learning algorithm, for quasi-ordered automata, and
have extended Goubault-Larrecq's generalized Valk-Jantzen lemma to finitely
generated monoids.

We have restricted ourselves to words on finite alphabets, but there are other
WQOs for which the generalized Valk-Jantzen lemma could reduce to an MAT-based
learning algorithm.
\begin{itemize}
\item Words on infinite (or very large) alphabets: this would require regrouping
  transitions in quasi-ordered automata by upward-closed subsets, so that they
  can still be finitely represented. This use of an additional structure to
  represent infinite alphabets finitely also comes up with \emph{nominal
    automata} \cite{bojanczykAutomataTheoryNominal2014}. More generally,
  enriched functors could provide a unifying framework for such state-machines
  on infinite alphabets: a Myhill-Nerode theorem for topos-enriched automata is
  notably given in \cite{iwaniackAutomataWToposesGeneral2024}.
\item Tree orderings: trees on well quasi-ordered families of terms can
  naturally be well quasi-ordered themselves
  \cite{kruskalWellQuasiOrderingTreeTheorem1960}, their ideals can be
  characterized as regular tree expressions \cite[\textsection
  10]{finkelForwardAnalysisWSTS2020}, but ``the technicalities are daunting''
  \cite{goubault-larrecqIdealApproachComputing2020}. Perhaps an approach based
  on the active learning algorithm for tree automata described in
  \cite{drewesLearningRegularTree2003} can soothe out these technicalities.
\end{itemize}

\bibliography{bibtex}

\begin{thebibliography}{10}

\bibitem{angluinLearningRegularSets1987}
Dana Angluin.
\newblock Learning regular sets from queries and counterexamples.
\newblock {\em Information and Computation}, 75(2):87--106, November 1987.
\newblock URL:
  \url{https://www.sciencedirect.com/science/article/pii/0890540187900526},
  \href {https://doi.org/10.1016/0890-5401(87)90052-6}
  {\path{doi:10.1016/0890-5401(87)90052-6}}.

\bibitem{aristoteActiveLearningDeterministic2024}
Quentin Aristote.
\newblock Active {{Learning}} of {{Deterministic Transducers}} with {{Outputs}}
  in {{Arbitrary Monoids}}.
\newblock In {\em 32nd {{EACSL Annual Conference}} on {{Computer Science
  Logic}} ({{CSL}} 2024)}, volume 288 of {\em Leibniz {{International
  Proceedings}} in {{Informatics}} ({{LIPIcs}})}, pages 11:1--11:20, Dagstuhl,
  Germany, 2024. Schloss-Dagstuhl - Leibniz Zentrum f{\"u}r Informatik.
\newblock URL:
  \url{https://drops.dagstuhl.de/entities/document/10.4230/LIPIcs.CSL.2024.11},
  \href {https://doi.org/10.4230/LIPIcs.CSL.2024.11}
  {\path{doi:10.4230/LIPIcs.CSL.2024.11}}.

\bibitem{barloccoCoalgebraLearningDuality2019}
Simone Barlocco, Clemens Kupke, and Jurriaan Rot.
\newblock Coalgebra {{Learning}} via {{Duality}}.
\newblock In Miko{\l}aj Boja{\'n}czyk and Alex Simpson, editors, {\em
  Foundations of {{Software Science}} and {{Computation Structures}}}, Lecture
  {{Notes}} in {{Computer Science}}, pages 62--79, Cham, 2019. Springer
  International Publishing.
\newblock \href {https://doi.org/10.1007/978-3-030-17127-8_4}
  {\path{doi:10.1007/978-3-030-17127-8_4}}.

\bibitem{bojanczykAutomataTheoryNominal2014}
Miko{\l}aj Boja{\'n}czyk, Bartek Klin, and S{\l}awomir Lasota.
\newblock Automata theory in nominal sets.
\newblock {\em Logical Methods in Computer Science}, Volume 10, Issue 3, August
  2014.
\newblock URL: \url{https://lmcs.episciences.org/1157}, \href
  {https://doi.org/10.2168/LMCS-10(3:4)2014}
  {\path{doi:10.2168/LMCS-10(3:4)2014}}.

\bibitem{colcombetAutomataMinimizationFunctorial2020}
Thomas Colcombet and Daniela Petrisan.
\newblock Automata {{Minimization}}: A {{Functorial Approach}}.
\newblock {\em Logical Methods in Computer Science}, Volume 16, Issue 1, March
  2020.
\newblock URL: \url{https://lmcs.episciences.org/6213/pdf}, \href
  {https://doi.org/10.23638/LMCS-16(1:32)2020}
  {\path{doi:10.23638/LMCS-16(1:32)2020}}.

\bibitem{colcombetLearningAutomataTransducers2021}
Thomas Colcombet, Daniela Petri{\c s}an, and Riccardo Stabile.
\newblock Learning {{Automata}} and {{Transducers}}: {{A Categorical
  Approach}}.
\newblock In Christel Baier and Jean {Goubault-Larrecq}, editors, {\em 29th
  {{EACSL Annual Conference}} on {{Computer Science Logic}} ({{CSL}} 2021)},
  volume 183 of {\em Leibniz {{International Proceedings}} in {{Informatics}}
  ({{LIPIcs}})}, pages 15:1--15:17, Dagstuhl, Germany, 2021. Schloss
  Dagstuhl--Leibniz-Zentrum f{\"u}r Informatik.
\newblock URL: \url{https://drops.dagstuhl.de/opus/volltexte/2021/13449}, \href
  {https://doi.org/10.4230/LIPIcs.CSL.2021.15}
  {\path{doi:10.4230/LIPIcs.CSL.2021.15}}.

\bibitem{drewesLearningRegularTree2003}
Frank Drewes and Johanna H{\"o}gberg.
\newblock Learning a {{Regular Tree Language}} from a {{Teacher}}.
\newblock In Zolt{\'a}n {\'E}sik and Zolt{\'a}n F{\"u}l{\"o}p, editors, {\em
  Developments in {{Language Theory}}}, Lecture {{Notes}} in {{Computer
  Science}}, pages 279--291, Berlin, Heidelberg, 2003. Springer.
\newblock \href {https://doi.org/10.1007/3-540-45007-6_22}
  {\path{doi:10.1007/3-540-45007-6_22}}.

\bibitem{finkelForwardAnalysisWSTS2020}
Alain Finkel and Jean {Goubault-Larrecq}.
\newblock Forward analysis for {{WSTS}}, part {{I}}: Completions.
\newblock {\em Mathematical Structures in Computer Science}, 30(7):752--832,
  August 2020.
\newblock URL:
  \url{https://www.cambridge.org/core/journals/mathematical-structures-in-computer-science/article/abs/forward-analysis-for-wsts-part-i-completions/E7183B3BE4D6B55364B06265E692638A},
  \href {https://doi.org/10.1017/S0960129520000195}
  {\path{doi:10.1017/S0960129520000195}}.

\bibitem{goubault-larrecqGeneralizationResultValk2009}
Jean {Goubault-Larrecq}.
\newblock On a {{Generalization}} of a {{Result}} by {{Valk}} and {{Jantzen}}.
\newblock Research Report LSV-09-09, {\'E}cole Normale Sup{\'e}rieure de
  Cachan, Cachan, May 2009.
\newblock URL: \url{https://hal.science/hal-03195658}.

\bibitem{goubault-larrecqIdealApproachComputing2020}
Jean {Goubault-Larrecq}, Simon Halfon, Prateek Karandikar, K.~Narayan Kumar,
  and Philippe Schnoebelen.
\newblock The {{Ideal Approach}} to {{Computing Closed Subsets}} in
  {{Well-Quasi-orderings}}.
\newblock In Peter~M. Schuster, Monika Seisenberger, Andreas Weiermann, Henrich
  Wansing, and Henrich Wansing, editors, {\em Well-{{Quasi Orders}} in
  {{Computation}}, {{Logic}}, {{Language}} and {{Reasoning}}: {{A Unifying
  Concept}} of {{Proof Theory}}, {{Automata Theory}}, {{Formal Languages}} and
  {{Descriptive Set Theory}}}, number~53 in Trends in {{Logic}}, pages 55--105.
  Springer Nature, Switzerland, 1 edition, 2020.
\newblock \href {https://doi.org/10.1007/978-3-030-30229-0_3}
  {\path{doi:10.1007/978-3-030-30229-0_3}}.

\bibitem{iwaniackAutomataWToposesGeneral2024}
Victor Iwaniack.
\newblock Automata in~{{W-Toposes}}, and~{{General Myhill-Nerode Theorems}}.
\newblock In Barbara K{\"o}nig and Henning Urbat, editors, {\em Coalgebraic
  {{Methods}} in {{Computer Science}}}, volume 14617 of {\em Lecture {{Notes}}
  in {{Computer Science}}}, pages 93--113, Cham, 2024. Springer Nature
  Switzerland.
\newblock \href {https://doi.org/10.1007/978-3-031-66438-0_5}
  {\path{doi:10.1007/978-3-031-66438-0_5}}.

\bibitem{jacobsAutomataLearningCategorical2014}
Bart Jacobs and Alexandra Silva.
\newblock Automata {{Learning}}: {{A Categorical Perspective}}.
\newblock In Franck {van Breugel}, Elham Kashefi, Catuscia Palamidessi, and Jan
  Rutten, editors, {\em Horizons of the {{Mind}}. {{A Tribute}} to {{Prakash
  Panangaden}}: {{Essays Dedicated}} to {{Prakash Panangaden}} on the
  {{Occasion}} of {{His}} 60th {{Birthday}}}, pages 384--406. Springer
  International Publishing, Cham, 2014.
\newblock \href {https://doi.org/10.1007/978-3-319-06880-0_20}
  {\path{doi:10.1007/978-3-319-06880-0_20}}.

\bibitem{kruskalWellQuasiOrderingTreeTheorem1960}
J.~B. Kruskal.
\newblock Well-{{Quasi-Ordering}}, {{The Tree Theorem}}, and {{Vazsonyi}}'s
  {{Conjecture}}.
\newblock {\em Transactions of the American Mathematical Society},
  95(2):210--225, 1960.
\newblock URL: \url{https://www.jstor.org/stable/1993287}, \href
  {https://arxiv.org/abs/1993287} {\path{arXiv:1993287}}, \href
  {https://doi.org/10.2307/1993287} {\path{doi:10.2307/1993287}}.

\bibitem{lodingAutomataFiniteTrees2021}
Christof L{\"o}ding and Wolfgang Thomas.
\newblock Automata on finite trees.
\newblock In Jean-{\'E}ric Pin, editor, {\em Handbook of {{Automata Theory}}},
  volume~1, pages 235--264. EMS Press, 1 edition, 2021.
\newblock URL: \url{https://doi.org/10.4171/Automata-1/7}, \href
  {https://doi.org/10.4171/AUTOMATA-1/7} {\path{doi:10.4171/AUTOMATA-1/7}}.

\bibitem{urbatAutomataLearningAlgebraic2020}
Henning Urbat and Lutz Schr{\"o}der.
\newblock Automata {{Learning}}: {{An Algebraic Approach}}.
\newblock In {\em Proceedings of the 35th {{Annual ACM}}/{{IEEE Symposium}} on
  {{Logic}} in {{Computer Science}}}, {{LICS}} '20, pages 900--914, New York,
  NY, USA, July 2020. Association for Computing Machinery.
\newblock URL: \url{https://dl.acm.org/doi/10.1145/3373718.3394775}, \href
  {https://doi.org/10.1145/3373718.3394775}
  {\path{doi:10.1145/3373718.3394775}}.

\bibitem{valkResidueVectorSets1985}
R{\"u}diger Valk and Matthias Jantzen.
\newblock The residue of vector sets with applications to decidability problems
  in {{Petri}} nets.
\newblock {\em Acta Informatica}, 21(6):643--674, March 1985.
\newblock \href {https://doi.org/10.1007/BF00289715}
  {\path{doi:10.1007/BF00289715}}.

\bibitem{vanheerdtLearningWeightedAutomata2020}
Gerco {van Heerdt}, Clemens Kupke, Jurriaan Rot, and Alexandra Silva.
\newblock Learning {{Weighted Automata}} over {{Principal Ideal Domains}}.
\newblock In Jean {Goubault-Larrecq} and Barbara K{\"o}nig, editors, {\em
  Foundations of {{Software Science}} and {{Computation Structures}}}, pages
  602--621, Cham, 2020. Springer International Publishing.
\newblock \href {https://doi.org/10.1007/978-3-030-45231-5_31}
  {\path{doi:10.1007/978-3-030-45231-5_31}}.

\bibitem{vanheerdtCALFCategoricalAutomata2020}
Gerrit~K. {van Heerdt}.
\newblock {\em {{CALF}} - Categorical Automata Learning Framework}.
\newblock PhD thesis, University College London, London, UK, 2020.
\newblock URL:
  \url{https://ethos.bl.uk/OrderDetails.do?uin=uk.bl.ethos.819869}.

\end{thebibliography}

\end{document}